\documentclass[pra,aps,superscriptaddress,showpacs,floatfix,tightenlines,twocolumn]{revtex4-1}
\usepackage{graphicx}
\usepackage{dcolumn}
\usepackage{bm}

\usepackage{amssymb}
\usepackage{amsmath}
\usepackage{url}
\usepackage{amsthm}
\usepackage{layout}
\usepackage{epsfig}
\usepackage{graphicx}
\usepackage{epstopdf}
\usepackage{booktabs}
\usepackage{float}
\usepackage{subfigure}
\usepackage[hyperindex,breaklinks]{hyperref}

\newtheorem{theorem}{Theorem}

\newtheorem{corollary}{Corollary}

\begin{document}

\title{Star network non-$n$-local correlations can resist consistency noises better}
\author{Kan He}
\email{hekanquantum@163.com}
\affiliation{College of Mathematics, Taiyuan University of Technology, Taiyuan, 030024, China}
%\affiliation {College of Information and Computer Science, Taiyuan University of Technology, Taiyuan, 030024, China}
%\affiliation {College of software, Taiyuan University of Technology, Taiyuan, 030024, China}
\author{Yueran Han}
\affiliation {College of Mathematics, Taiyuan University of Technology, Taiyuan, 030024, China}

\date{\today }

%\thanks{{\it 2002 Mathematical Subject Classification.} 47H20, 47B49, 47A12}
%\thanks{{\it Key words and phrases.}  Gaussian coherence breaking channels, assistant entanglement inputs}
%\thanks{This work supported by National
%Science Foundation of China (11201329, 11171249) and Program for the Outstanding Innovative Teams of Higher Learning Institutions of Shanxi.}

\begin{abstract}
Imperfections from devices can result in the decay or even vanish of non-$n$-local correlations as the number of sources $n$ increases in the polygon and linear quantum networks. Even though so does this phenomenon for the very special kind of noises, consistency noises of a sequence of devices, which means the sequence of devices have the same probability fails to detect. However, non-multi-local correlations in different quantum networks display distinct anti-noise powers. In the paper, we discover that star network quantum non-$n$-local correlations can resist consistency noises best among noncyclic networks ones. We first calculate the noisy  star network non-$n$-local inequality criteria and analyze its persistency conditions theoretically by the violation of these inequalities. Based on these criteria, we find that in a star network,  the persistency number of sources $n$ has the ability to be rid of consistent noises, and approximates to the infinity. Polygon and linear network  non-$n$-local correlations can not meet the requirements. More significantly, we prove mathematically star network non-$n$-local correlations have the stronger anti-consistent-noise power than arbitrary noncyclic network one. Finally, more generally, we introduce the topic of  partially consistent noises and  compare the change pattern of finite persistency of star network non-$n$-local correlations with linear network ones.
\end{abstract}

\maketitle

%%%%%%%%%%%%%%%%%%%%%%%%%%%%%%%%%%%%%%%%%%%%%%%%%%%%%%%%%%%%%%%%%%%

\section{Introduction}

Quantum nonlocality, as an important resource in quantum information processing, has many applications (\cite{Sca}, \cite{WJD}). It offers
an advantage in communication complexity
problems \cite{CB}, device-independent quantum cryptography
\cite{BHK}, \cite{MPA}, randomness expansion \cite{Pir}, \cite{CK}, and measurementbased
quantum computation \cite{RB}, \cite{RBB}. Quantum nonlocality
can be demonstrated by violation of Bell type inequalities theoretically and experimentally
 \cite{BCP}-\cite{RKM}.

A quantum network refers to multiple parties
connected by quantum sources, and dilates by increasing the new sources \cite{WEH}. The theory
of quantum networks arises from the large-scale
quantum communication  \cite{23, 24}.
Subsequentially, the topic of network quantum correlations is proposed and has been focused on, in particular network quantum nonlocality \cite{25}-\cite{27}. A  number of
inequality criteria for nonlocality in various quantum networks
have been found, such as entanglement-swapping
networks \cite{5, 6}, linear networks \cite{8}, star networks
\cite{9,10}, polygon networks \cite{11}-\cite{13}, tree-shaped networks
\cite{14}-\cite{16} and other networks \cite{17}.

Sequentially, a challenging task is discover the difference of non-multi-local correlations between different-structure quantum networks.  As we know, nonlocal correlations in distinct kinds of quantum networks have different definitions and inequality criteria. However, it is not clear that how different they are and how the network structures result in the differences. Or more specially, one may want to ask, if there are  a linear and star network with the same 6 two qubit state as their sources, what is the difference between the two kinds of non-6-local correlations? In the paper, we want to  answer the above questions from  the perspective of anti-noise powers of different network nonlocal correlations. Indeed,  in a noisy scenario, network nonlocal correlations will decay as the network dilates, where noises maybe come from  entanglement generation, communications
over noisy quantum channels and imperfections
in measurements \cite{trinoisy, linearnoisy}. Considering these noises, Mukherjee et al showed that how the persistency of  non-$n$-local correlations in the linear and polygon networks were constrained under the influence of noises \cite{trinoisy, linearnoisy}.  Furthermore, we also observe that the noisy sensitivity of non-$n$-local correlations is different between the polygon and linear quantum networks.
In the paper, we conclude that the noisy star network non-$n$-local correlation has stronger persistency ability than the linear network one, even an arbitrary noncyclic network one. This means that star-type structure can offer the more anti-noise network non-multi-local correlation.
Recall the star network is
a primary motif of neural networks, which is important in a human-engineered computer network configuration \cite{zsnet}. Moreover, considering the more general partially consistent noises, we also discover the difference between noisy persistency of linear and star network non-$n$-local correlations.

This paper is structured as follows. In Section II, we review the topics of non-multi-local correlations in the star network and noise generation. In
Section III, we calculate the noisy star network local inequality criteria and obtain persistency  conditions of non-$n$-local correlations in the star network under influence of noises. Furthermore, we discover that the star network non-$n$-local correlations can be demonstrated infinitely under the consistent noises and plot the corresponding noise region. It is also proved that star network non-$n$-local correlations have the stronger persistency power than arbitrary noncyclic network one. Section IV, we propose the topic of partially consistent noises and analyze the change pattern of the maximal number of sources $n_{\rm max}$, which satisfies noisy star network non-$n$-local correlations can be demonstrated for $n\leq n_{\rm max}$.

\section{Preliminaries}\label{sec:2}

We recall some necessary notations  in the section.

\subsection{Matrix representation of a two-qubit state}

An arbitrary two-qubit state $\varrho$ has the following matrix representation,
\begin{align*}
\varrho =& \frac{1}{4}( \mathbb{I}_2\times \mathbb{I}_2+\vec a \vec\sigma\otimes \mathbb{I}_2 + \mathbb{I}_2 \otimes \vec b \vec \sigma\\
 & +\sum\limits_{{j_1}, {j_2} = 1}^3 {{\omega _{{j_1}{j_2}}}{\sigma _{{j_1}}} \otimes {\sigma _{{j_2}}}} )
\end{align*}
where $\vec \sigma=(\sigma_1,\sigma_2,\sigma_3)$, $\sigma _{j_k}$ are Pauli matrices,  $(j_k=1,2,3)$. $\vec a=(x_1,x_2,x_3)$ and $\vec b=(y_1,y_2,y_3)$ are local Bloch real vectors with $|\vec a|,|\vec b|\leq 1$ and $T_{\varrho} =(\omega_{j_1,j_2})_{3\times3}$ denotes  real correlation tensor, $\omega _{{j_1},{j_2}}=Tr[\rho\sigma_{j_1}\otimes\sigma_{j_2}]$. As $T_{\varrho}$ can be diagonalized, a simplified expression of $\varrho$ is
\begin{align*}
\varrho'=\frac{1}{4}( {\mathbb{I}_{2}\times \mathbb{I}_{2}+\vec {\mathfrak{a}} \vec \sigma   \otimes \mathbb{I}_2 + \mathbb{I}_2 \otimes \vec {\mathfrak{b}} \vec \sigma   + \sum\limits_{j= 1}^3 {{t _{jj}}{\sigma _{j}} \otimes {\sigma _{j}}}}).
\end{align*}
Where $ T={\rm diag}(t_{11},t_{22},t_{33})$ is the correlation matrix with $t_{11},t_{22},t_{33}$ being the eigenvalues of $\sqrt{T_{\varrho}^TT_{\varrho}}$, i.e., singular values of $T_{\varrho}$.

\subsection{Non-$n$-local correlations in star-shaped networks}

A star-shaped network composes of $n + 1$ parties ($n$ sources), where a central node (referred to as Bob) shares a bipartite state with each of the $n$ nodes (referred to as Alices) (see Fig. 1). Where bipartite states are provided by $n$ independent sources. Here assume each of the $n$ Alices performs dichotomic measurements with two outputs. The inputs are denoted by $x_i \in \{0, 1\}$ for the  $i$th Alice and the outcomes are denoted by $a_i \in \{0, 1\}$. Bob measures one input $y$ with 2 outputs labeled by $b_j \in \{0, 1\}$. The correlations in the star network is characterized by the probability decomposition \cite{9}
\begin{align*}
& p (\{a_i\}_{i=1,...,n},b|\{x_i\}_{i=1}^n,y)\\
& =\int\left( \Pi_{i = 1}^n d\lambda_i p(\lambda_i)p(a_i|x_i,\lambda_i)  \right) p(b|y,\{\lambda_i\}_{i=1,...,n}).
\end{align*}

\begin{figure}[h] \label{startu}
\centering
\includegraphics[width=3.5in]{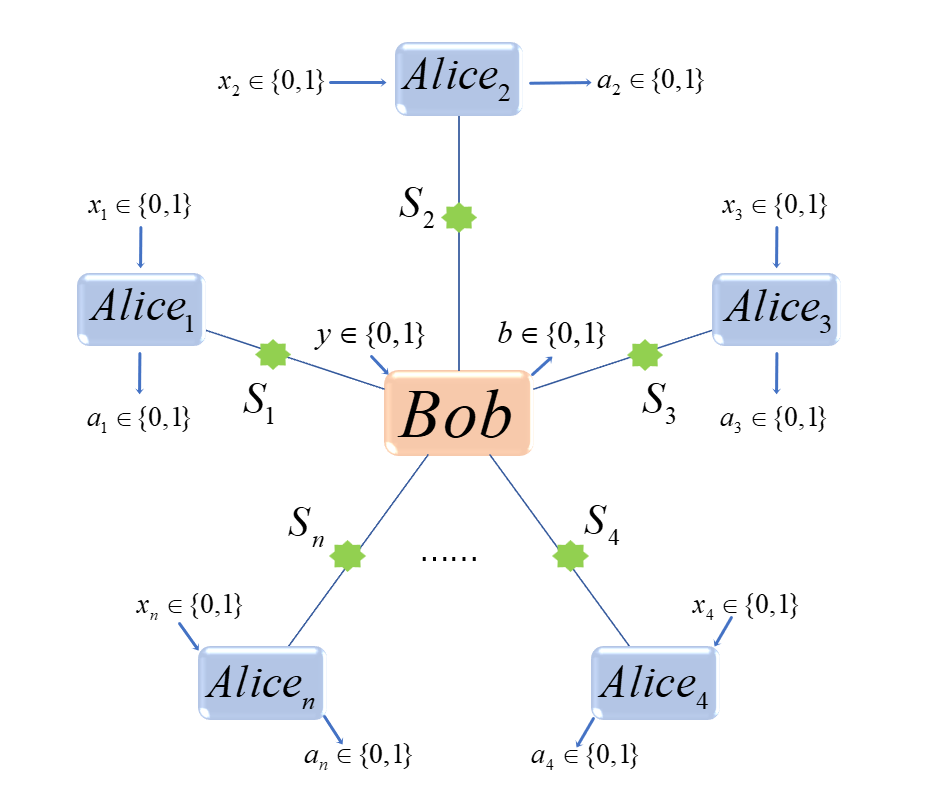}\\
\caption{\quad The star network scenario with $n+1$ parties}
\centering
\end{figure}
This implies the following $n$-local inequality criterion holds ture \cite{9}
\begin{equation}\label{e1}
\mathcal{S}_{\rm star}=|I|^{1/n}+|J|^{1/n}\leqslant 1 \end{equation}
where
\begin{equation*}
\begin{aligned} I&=\frac{1}{2^n}\sum\limits_{{x_1}...{x_n}} {\left\langle A^{1}_{x_1} ...A^{n}_{x_n}B_0 \right\rangle },\\
J&=\frac{1}{2^n}\sum\limits_{{x_1}...{x_n}}(-1)^{\sum\nolimits_i x_i} {\left\langle A^{1}_{x_1} ...A^{n}_{x_n}B_1 \right\rangle }\\
&{\left\langle A^{1}_{x_1} ...A^{n}_{x_n}B_y \right\rangle }= \\
&\sum\limits_{{a_1}...{a_n},b}(-1)^{b+\sum\nolimits_i a_i}p(\{a_i\}_{i=1,..,n},b|\{x_i\}_{i=1,..,n},y)
  \end{aligned}
\end{equation*}
Violation of Eq. (\ref{e1}) demonstrates the  non-$n$-local correlation in the star network.

When $\mathcal{S}_i(i = 1, 2,..., n)$ generate an arbitrary two qubit state $\varrho_i=\varrho_{A_iB}$, each  $A_i(i = 2, 3,..., n)$  receives one qubit of  $\varrho_{i}$. Bob  receives $n$ qubits. Let each of the Alice parties perform projection
in the Bell basis, referred to as a Bell state measurement. The inequality criterion (\ref{e1}) becomes
\begin{align}\label{qubitiq}
\mathcal{S}_{\rm star}=\sqrt{\left (\Pi_{i=1}^n t_1^{A_i} \right )^{1/n}+\left (\Pi_{i=1}^n t_2^{A_i} \right )^{1/n}}\leqslant 1
\end{align}
where $t_1^{A_i}$and $t_2^{A_i}$are the two greatest eigenvalues of the matrix $T_{\varrho_{{A_i} B}}^T T_{\varrho_{{A_i} B}}$ with $t_1^{A_i} \geq t_2^{A_i}$ \cite{10}.

\subsection{Noise generation}

The following kinds of noises will be considered in the paper, including imperfection in measurements, errors from entanglement generation and noises from communications.

{\bf Imperfection in measurements} Let $\beta_i \in  [0, 1]$ characterize imperfection in the measurement operator $M_i=M_i^{ideal}$, which means it fails to detect with probability $1 - \beta_i$ \cite{linearnoisy}. We call the corresponding noise is {\it consistent} if $\beta_i=\beta$ for all $i$ and constant $\beta$.  The corresponding noisy measurement operator  is denoted by $M_i^{noisy}$. Thus
\begin{eqnarray*}
M^{noisy}_{ik_{i}0}&=&\beta_i \mathcal{Q}_i^+ +\frac{1-\beta_i}{2}\mathbb{I}_{2\times2},\\
M^{noisy}_{ik_{i}1}&=&\beta_i \mathcal{Q}_i^- +\frac{1-\beta_i}{2}\mathbb{I}_{2\times2},
\forall i=1,2,...,n, \ k_i=0,1.
\end{eqnarray*}
So
\begin{eqnarray}\label{e2}
M^{noisy}_{0}&=&\mathop  \otimes \limits_{i = 1}^n M^{noisy}_{ik_{i}0},\\
M^{noisy}_{1}&=&\mathop  \otimes \limits_{i = 1}^n M^{noisy}_{ik_{i}1},
\forall i=1,2,...,n.k_i=0,1.
\end{eqnarray}
Where $\mathcal{Q}_i^+ (\mathcal{Q}_i^-)$ denotes the projection operator corresponding to the +1 (-1) eigenvalue, namely,  projections corresponding to perfect projective measurement. Denote  by $M^{noisy}_{ik_{i}0},M^{noisy}_{ik_{i}1}$   the projections corresponding to perfect measurement $\vec b_{k_1}.\vec \sigma$ to $B_{k_i}$ which is the Bob's component corresponding to the $i$th Alice, we obtain
\begin{eqnarray} \label{imm1}
M^{noisy}_{ik_{i}0}&=& \beta_i M^{ideal}_{ik_{i}0} +\frac{1-\beta_i}{2}\mathbb{I}_{2\times2},\\
M^{noisy}_{ik_{i}1}&=& \beta_i M^{ideal}_{ik_{i}0} +\frac{1-\beta_i}{2}\mathbb{I}_{2\times2}.
\end{eqnarray}

Similar,  for party $A_i$, let $\mu_i \in [0,1]$ parametrize a faulty measurement device. It means that  for single-qubit projection, such a device fails to detect any output with probability $1 - \mu_i$. POVM resulting due to imperfection in $\vec a_{k_i}.\vec \sigma$ to $A_{k_i}$ thus has two elements $\{P^{noisy}_{ik_i 0},P^{noisy}_{ik_i 1}\}$ given by
\begin{eqnarray*}
P^{noisy}_{ik_{i}0}&=&\mu_i \mathcal{O}_i^+ +\frac{1-\mu_i}{2}\mathbb{I}_{2\times2},\\
P^{noisy}_{ik_{i}1}&=&\mu_i \mathcal{O}_i^- +\frac{1-\mu_i}{2}\mathbb{I}_{2\times2},
\forall i=1,2,...,n.k_i=0,1.
\end{eqnarray*}
\iffalse Where $\mathcal{O}_i^+ (\mathcal{O}_i^-)$ denotes the projection operator corresponding to the +1 (-1) eigenvalue. $\mathcal{O}_i^+ (\mathcal{O}_i^-)$ denotes projectors corresponding to perfect projective measurement. Let   $P^{noisy}_{ik_{i}0},P^{noisy}_{ik_{i}1}$  be the projectors corresponding to perfect measurement $\vec a_{k_1}.\vec \sigma$ to $A_{k_i}$,\fi Write similarly
\begin{eqnarray}\label{imm3}
P^{noisy}_{ik_{i}0}&=&\mu_i P^{ideal}_{ik_{i}0} +\frac{1-\mu_i}{2}\mathbb{I}_{2\times2},\\ \label{imm4}
P^{noisy}_{ik_{i}1}&=&\mu_i P^{ideal}_{ik_{i}1} +\frac{1-\mu_i}{2}\mathbb{I}_{2\times2}.
\end{eqnarray}

{\bf Errors from entanglement generation}
An ideal entangled pure state is generated by acting on $|10\rangle$ with  Hadamard and CNOT gates.
However, in practical situations, imperfections from  preparation devices results in a mixed entangled state \cite{errorentangle}. Such errors come from applications of Hadamard and CNOT gates. In each source $\mathcal{S}_i$, let $\alpha_i$ and $\delta_i$ denote the imperfection parameters characterizing $\mathcal{H} $ and CNOT gates, respectively. For $i = 1, 2,..., n$, starting from $\varrho_i=|10\rangle \langle10|$, and the noisy Hadamard gate generates
\begin{align*}
\varrho'_i
&=\alpha_i(\mathcal{H} \otimes \mathbb{I}_2 \varrho_i \mathcal{H}^\dagger \otimes \mathbb{I}_2)+\frac{1-\alpha_i}{2} \mathbb{I}_2 \otimes \varrho_{2i}, (\varrho _{2i}=Tr(\varrho_i))\\
&=\frac{1}{2}(|00\rangle \langle00|+|10\rangle \langle10|)-\frac{\alpha_i}{2}(|00\rangle \langle10|+|10\rangle \langle00|)
\end{align*}
Takeing $\varrho'_i$ to noisy CNOT gives
\begin{align}\label{ns1}
\varrho ^{''}_i
=&\delta_i({\rm CNOT} \varrho'_i ({\rm CNOT})^\dagger )+\frac{1-\delta_i}{4} \mathbb{I}_2 \otimes \mathbb{I}_2 \nonumber\\
=&\frac{1}{4}\{\sum\limits_{k,j = 0}^1 {[1 + {{( - 1)}^{k + j}} \delta_i |kj\rangle \langle kj|]}\nonumber \\
&-2\alpha_i \delta_i (|11\rangle \langle00|+|00\rangle \langle11|) \}
\end{align}
The correlation tensor of $\varrho ^{''}_i$ is diag$(-\alpha_i \delta_i,\alpha_i \delta_i,\delta_i)$. Here we also call the error from Hadamard/CNOT gates is {\it consistent} if $\alpha_i$s or $ \delta_i$s equal to a constant.

{\bf Amplitude-damping (AD) and phase-damping (PD) channels}
For $ i=1,2,...,n $, let $\gamma_i^{amp}$ and $\xi_i^{amp}$ characterize amplitude-damping channels connecting $\mathcal{S}_i$ with $A_{i}$ and $B$, respectively. The amplitude-damping channel (e.g., parametrized by $\gamma_i^{amp}$) is represented by Krauss operators $|0\rangle \langle0|+\sqrt{1-\gamma^{amp}}|1\rangle \langle1|$ and$\sqrt{\gamma^{amp}}|0\rangle \langle1|$.
We also call the noise parameter from the amplitude-damping channel is {\it consistent} if $\gamma^{amp}$ is always a common constant. When $\gamma^{amp}=0$, the noise vanish. Similarly  let $\gamma_i^{ph},\xi_i^{ph}$ characterize channels connecting $\mathcal{S}_i$ with $A_{i}$ and $B$, respectively \cite{linearnoisy}.  Any amplitude-damping channel (e.g., parametrized by $\gamma_i^{ph}$) is represented by Krauss operators $|0\rangle \langle0|+\sqrt{1-\gamma^{ph}}|1\rangle \langle1|$ and$\sqrt{\gamma^{ph}}|1\rangle \langle1|$.

\section{Power of noisy persistency of star network non-$n$-local correlations}

In the section, we discover and analyze the stronger power of noisy persistency of star network non-$n$-local correlations.

\subsection{Inequality criteria}

In the following theorem, we consider the general case that each source is an abitrary two qubit state and measurements are imperfect.

\begin{theorem} \label{th1}
Assume each source $\mathcal{S}_i$ generating an arbitrary two-qubit state and all the parties performing imperfect measurements (\ref{imm1})- (\ref{imm4}), respectively. Then the noisy star network non-$n$-local correlations are demonstrated if
\begin{align}\label{max}
& \mathcal{S}_{\rm star}^{\rm noisy}= \nonumber\\ &(\Pi_{i=1}^n\mu_i\beta_i)^{1/n}\sqrt{(\Pi_{i=1}^nt_1^{A_i})^{1/n} +(\Pi_{i=1}^nt_2^{A_i})^{1/n} }>1.
\end{align}
where $t_1^{A_i}$and $t_2^{A_i}$are the two greatest positive eigenvalues of the matrix $T_{\varrho_{{A_i} B}}^T T_{\varrho_{{A_i} B}}$ with $t_1^{A_i} \geq t_2^{A_i}$.
\end{theorem}

%%\langle O_1 \otimes O_2 \rangle _{\rho} =Tr \left [(O_1\otimes O_2 )\rho \right ]=Tr \left [\Sigma_{j,k=1,2,3}(v_1^j v_2^k \sigma_j \otimes \sigma_k)\rho \right ]=\Sigma_{j,k=1,2,3} v_1^j v_2^k t_{jk} =\vec v_1(T_{\rho} \vec v_2)
\begin{proof}
See the Appendix.
\end{proof}

 When errors from entanglement generation occur like Eq. (\ref{ns1}), the $i$th source generates the noisy state
$\varrho ^{''}_i$ with $ T_{\varrho ^{''}_i}^*T_{\varrho ^{''}_i}={\rm diag}(\alpha_i^2 \delta_i^2,\alpha_i^2 \delta_i^2,\delta_i^2)$.
So we obtain the following corollary from Theorem \ref{th1}.

\begin{corollary} \label{nsco}
Assume each source $\mathcal{S}_i$ generate the noisy two-qubit state $\varrho ^{''}_i$ in Eq. (\ref{ns1}),  and all the parties performing imperfect measurements (\ref{imm1})- (\ref{imm4}), respectively. Then the noisy star network non-$n$-local correlations are demonstrated if
\begin{align}\label{nsv}
\mathcal{S}_{\rm star}^{\rm noisy}=\left(\Pi_{i = 1}^n {\beta_i \mu_i \delta_i} \right)^{1/n} \sqrt{\Pi_{i=1}^n \alpha_i^{2/n}+1 }>1.
\end{align}
\end{corollary}

Furthermore, when each noisy two-qubit state $\varrho ^{''}_i$ is sent by  amplitude-damping channels parameterized by $\gamma_i^{amp}$ and $\xi_i^{amp}$,  the output state  $\varrho ^{'''}_i$ is with $T_{\varrho ^{'''}_i}^*T_{\varrho ^{'''}_i}={\rm diag}(\alpha_i^2 \delta_i^2 D_i^{amp},\alpha_i^2 \delta_i^2 D_i^{amp},(\delta_i D_i^{amp}+\gamma_i^{amp}\xi_i^{amp})^2)$, where $D_i^{amp}=(1-\gamma^{amp}_i)(1-\xi^{amp}_i)$ \cite{linearnoisy}.

\begin{corollary} \label{channel1}
Assume each source $\mathcal{S}_i$ generates the noisy two-qubit state $\varrho ^{''}_i$, $\gamma_i^{amp}$ and $\xi_i^{amp}$ characterize amplitude-damping channels connecting $\mathcal{S}_i$ with $\mathcal{A}_{i}$ and $B$ respectively,  and all the parties performing imperfect measurements (\ref{imm1})- (\ref{imm4}), respectively. Then the noisy star network non-$n$-local correlations are demonstrated if
\begin{equation}
\label{c1v}
\mathcal{S}_{\rm star}^{\rm noisy}=\left(\Pi_{i = 1}^n {\beta_i \mu_i } \right)^{1/n}\sqrt{{\rm Max}(F_1,F_2)}>1,
\end{equation}
where
\begin{equation*}
  \begin{aligned}
F_1=&2\left(\Pi_{i = 1}^n \alpha_i^2 \delta_i^2 D_i^{amp}\right)^{1/n},\\
F_2=&\left(\Pi_{i = 1}^n \alpha_i^2 \delta_i^2 D_i^{amp}\right)^{1/n}\\
&+\Pi_{i = 1}^n(\delta_i D_i^{amp}+\gamma_i^{amp}\xi_i^{amp})^{2/n}
  \end{aligned}
\end{equation*}
\end{corollary}

Instead of amplitude-damping channels, we consider the noises from phase-damping channels parameterized by $\gamma_i^{ph}$ and $\xi_i^{ph}$. Here the output state  $\varrho ^{'''}_i$ is with $T_{\varrho ^{'''}_i}^*T_{\varrho ^{'''}_i}={\rm diag}(\alpha_i^2 \delta_i^2 D_i^{ph},\alpha_i^2 \delta_i^2 D_i^{ph},\delta_i^2)$, where $D_i^{ph}=(1-\gamma^{ph}_i)(1-\xi^{ph}_i)$ \cite{linearnoisy}.

\begin{corollary} \label{channel2}
Assume each source $\mathcal{S}_i$ generates the noisy two-qubit state $\varrho ^{''}_i$, $\gamma_i^{ph}$ and $\xi_i^{ph}$ characterize amplitude-damping channels connecting $\mathcal{S}_i$ with $\mathcal{A}_{i}$ and $B$ respectively,  and all the parties performing imperfect measurements (\ref{imm1})- (\ref{imm4}). Then the noisy star network non-$n$-local correlations are demonstrated if
\begin{align}\label{channel2v}
& \mathcal{S}_{\rm star}^{\rm noisy}= \nonumber\\ &\left(\Pi_{i = 1}^n {\beta_i \mu_i } \right)^{1/n}\sqrt{\left(\Pi_{i = 1}^n \alpha_i^2 \delta_i^2 D_i^{ph}\right)^{1/n}+\left(\Pi_{i = 1}^n \delta_i \right)^{2/n}}>1.
\end{align}

\end{corollary}

An amazing and exclusive-to-star-network characteristic is observed that the network dilation index $n$ will vanish from formulas \ref{nsv}, \ref{c1v} and \ref{channel2v} when the corresponding noises are consistent. Indeed, assume that $\alpha_i=\alpha, \delta_i=\delta, \mu_i=\mu,\beta_i=\beta$, $\forall i = 1, 2,..., n$.
From Corollary \ref{nsco}, we have the following conclusion.

\begin{corollary} \label{cnsco}
With the same assumption in Corollary \ref{nsco}, and each kind of noises is consistent. Then the noisy star network non-$n$-local correlations are demonstrated if
\begin{align}\label{cnsv}
\mathcal{S}_{\rm star}^{\rm noisy}=\beta \mu \delta  \sqrt{\alpha^2 +1 }>1.
\end{align}
\end{corollary}

The corollary says that the star network non-$n$-locality will be demonstrated for arbitrary $n$ if the noisy parameters satisfy $\beta \mu \delta  \sqrt{\alpha^2 +1 }>1$. The polygon and linear networks do not meet the requirement, where $n$ is always bounded not to be infinite  because of constraint from the noisy parameters \cite{trinoisy} and \cite{linearnoisy}. Similarly, corresponding to Corollary \ref{channel1} and \ref{channel2}, we have the following corollaries.

\begin{corollary} \label{cchannel1}
With the same assumptions with Corollary \ref{channel1}, and each kind of noises is consistent.  Then the noisy star network non-$n$-local correlations are demonstrated if
\begin{align}
\label{cchannel1v}
& \mathcal{S}_{\rm star}^{\rm noisy}=\nonumber\\ & \beta\mu \sqrt{{\rm Max}(2\alpha^2 \delta^2 D^{amp},\alpha^2 \delta^2D^{amp}+(\delta D^{amp}+\gamma^{amp}\xi^{amp})^2)}\nonumber\\ & >1,
\end{align}
where $
  D^{amp}=(1-\gamma^{amp})(1-\xi^{amp})$.
\end{corollary}

\begin{corollary} \label{cchannel2}
With the same assumptions with Corollary \ref{channel2}, and each kind of noises is consistent.  Then the noisy star network non-$n$-local correlations are demonstrated if
\begin{align}\label{cchannel2v}
\mathcal{S}_{\rm star}^{\rm noisy}=\beta \mu \delta\sqrt{ \alpha^2(1-\gamma^{ph})(1-\xi^{ph})+1}>1.
\end{align}

\end{corollary}

We observe that $n$ also vanish from (\ref{cchannel1v}) and (\ref{cchannel2v}). In summary, we may denote by the fixed set notation $\Omega^{\rm noisy}(v_1, v_2, ...,v_k)$ ($v_i$ is the $i$th kind of the noisy parameters)   the corresponding region of noisy parameters such that the star network non-$n$-locality can be detected  for arbitrary $n$. Also one
can see from Corollaries \ref{cnsco}, \ref{cchannel1} and \ref{cchannel2} that different consistency noisy parameters make distinguished effects on the set $\Omega^{\rm noisy}(v_1, v_2, ...,v_k)$. In the following section, we devote to proving that star network nonlocal correlation is most anti-noise among noncyclic network ones, and discovering the change pattern of the set  $\Omega^{\rm noisy}(v_1, v_2, ...,v_k)$ under different combinations of noises,  including the single and mixed noises. While, comparing with results in linear network case \cite{linearnoisy}, we explore the difference between the two network nonlocal correlations.

\subsection{Anti-noise power of star network non-$n$-local correlations}

Here, we first observe that  star network non-$n$-local correlations are strongest  among noncyclic networks with the same number of sources. Indeed,  With the same assumption with Theorem 1 for an arbitrary  noncyclic network with $n$ independent sources,  by the calculation similar to the proof of Theorem 1, it follows from \cite{15} that the noisy noncyclic network non-$n$-local correlations are demonstrated if
\begin{align}\label{max}
\nonumber
& \mathcal{S}_{\rm noncyclic}^{\rm noisy}= \nonumber\\ &(\Pi_{i=1}^n\mu_i\beta_i)^{1/p_n}\sqrt{(\Pi_{i=1}^nt_1^{A_i})^{1/p_n} +(\Pi_{i=1}^nt_2^{A_i})^{1/p_n} }>1.\nonumber
\end{align}
where $p_n$ denotes the source number in the last layer of this noncyclic network, $t_1^{A_i}$ and $t_2^{A_i}$ are the two greatest positive eigenvalues of the matrix $T_{\varrho_{{A_i} B}}^T T_{\varrho_{{A_i} B}}$ with $t_1^{A_i} \geq t_2^{A_i}$.
Note that each $t_1^{A_i}\leq 1$ and $p_n\leq n$. Comparing to Ineq. (3.1), we have
\begin{align}
&\mathcal{S}_{\rm noncyclic}^{\rm noisy}\nonumber\\ & =(\Pi_{i=1}^n\mu_i\beta_i)^{1/p_n}\sqrt{(\Pi_{i=1}^nt_1^{A_i})^{1/p_n} +(\Pi_{i=1}^nt_2^{A_i})^{1/p_n} } \nonumber\\ & \leq(\Pi_{i=1}^n\mu_i\beta_i)^{1/n}\sqrt{(\Pi_{i=1}^nt_1^{A_i})^{1/n} +(\Pi_{i=1}^nt_2^{A_i})^{1/n} }\nonumber\\ & =\mathcal{S}_{\rm star}^{\rm noisy}.
\end{align}
 Namely,  $\mathcal{S}_{\rm noncyclic}^{\rm noisy}>1$ implies $\mathcal{S}_{\rm star}^{\rm noisy}>1$. So we conclude the following theorem.

\begin{theorem}
Assume $\mathcal{A}$ is an arbitrary noncyclic network with parties $A_i$s, and $\mathcal{B}$ is a star network  with parties $B_j$s. All source noisy states of $\mathcal{A}$ and $\mathcal{B}$ are the same. $A_k$ and $B_k$ perform the same noisy measurements if $i=j=k$. Denote $n$ by the source number of $\mathcal{A}$ ($\mathcal{B}$). Then  under the same noisy level, the non-$n$-local correlation is demonstrated in $\mathcal{B}$ if it is demonstrated in $\mathcal{A}$.
\end{theorem}

By Theorem 2, the star network offer the most powerful persistency of noncyclic netwrok non-multi-local correlation. For instance, one can calculate that under the noisy level $(\mu=\beta=\nu,\gamma,\xi,\alpha,\delta)=(0.98,0.05,0.05,0.95,0.97)$, non-8-local correlation at most can be detected in the linear network, but the non-$n$-local correlation can always be detected for arbitrary $n$ in the star network.

Next we show the change patterns of the infinite persistency sets under different noises and comparison with the linear network case.

{\bf Case 1.} (single kind of consistent noises only) $\mu_i=\mu, \beta_i=\beta$, other noises vanish, namely, $\alpha_i=1,  \delta_i=1,  \xi_i^{amp}=0, \eta_i^{amp}=0, \xi_i^{ph}=0, \eta_i^{ph}=0$, $\forall i = 1, 2,..., n$.
The inequality
criterion (\ref{cnsv}) becomes $\sqrt{2}\mu\beta >1$. The corresponding infinite-persistency set $\Omega^{\rm noisy}(\mu, \beta)=\{\sqrt{2}\mu\beta>1\}$ is  as the blue region plotted in Fig. 2 (a). Comparing with the linear network case \cite{linearnoisy}, under the noisy level $(\mu, \beta)=(0.95, 0.95)$, the persistency source number is 12 at most in the linear network, and here $+\infty $ in the star network.

When only single kind of consistent state noises exists, $\alpha_i=\alpha,\delta_i=\delta$, $\forall i = 1, 2,..., n$, other noises vanish.   The corresponding infinite persistency set $\Omega^{\rm noisy}(\alpha, \delta)$ is plotted in Fig. 2 (b). One can calculate that under the state noisy level $(\alpha, \delta)=(0.91, 0.85)\in \Omega^{\rm noisy}(\alpha, \delta)$, only the linear network non-$3$-local correlation can be demonstrate.  Similarly, Fig. 2 (c) explores the region of the infinite persistency set $\Omega^{\rm noisy}(\gamma^{amp}, \xi^{amp})$. Here, under the noisy parameter pair $(0.25,0.27)\in \Omega^{\rm noisy}(\gamma^{amp}, \xi^{amp})$, a common calculation show the non-$2$-local correlation in the linear network can be detected.

When there is only consistent PD channel noises in the star network, namely,  $\gamma_i^{ph}=\gamma^{ph}$,  $\xi_i^{ph}=\xi^{ph}$, $\forall i = 1, 2,..., n$, other noisy parameters vanish. The inequality criterion \ref{cchannel2v} becomes $\sqrt{(1-\gamma^{ph})(1-\xi^{ph})+1}>1$, which always holds true for all parameters $\gamma^{ph}\in [0,1], \xi^{ph}\in [0,1]$. This says that the single phase damping channel noises can not decay the non-$n$-local correlations in the star network.

\begin{figure}\label{singlenoise}
\center
\subfigure []
{\includegraphics[width=3.9cm,height=2.9cm]{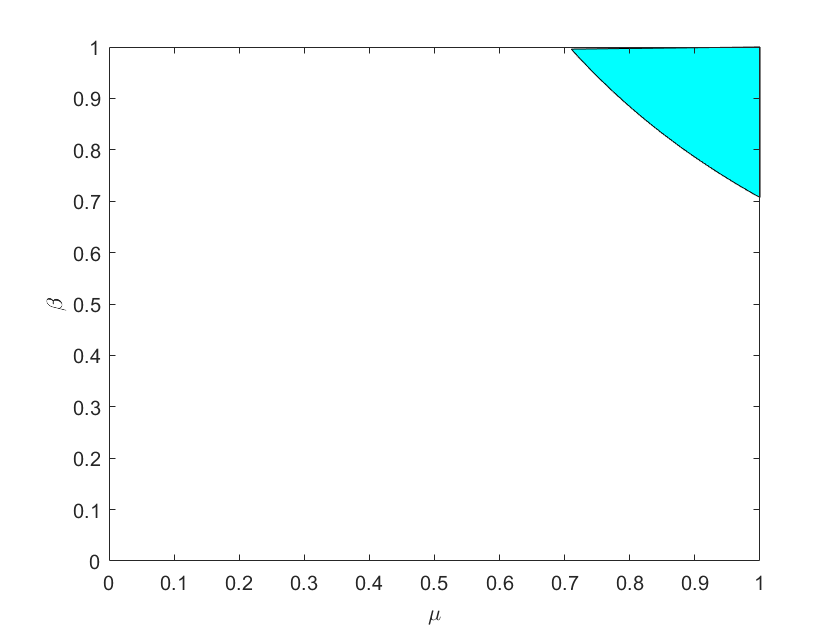}}
\subfigure []
{\includegraphics[width=3.9cm,height=2.9cm]{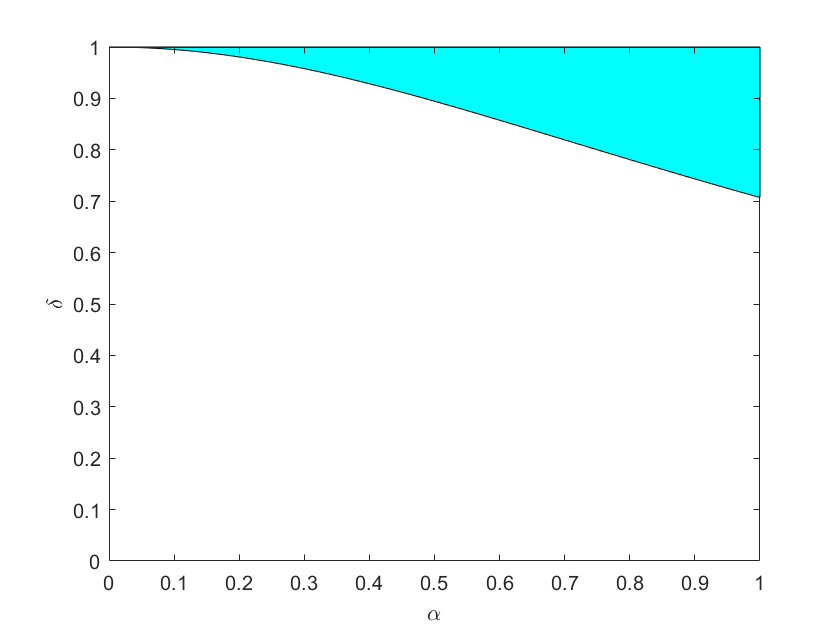}}\\
\subfigure []
{\includegraphics[width=3.9cm,height=2.9cm]{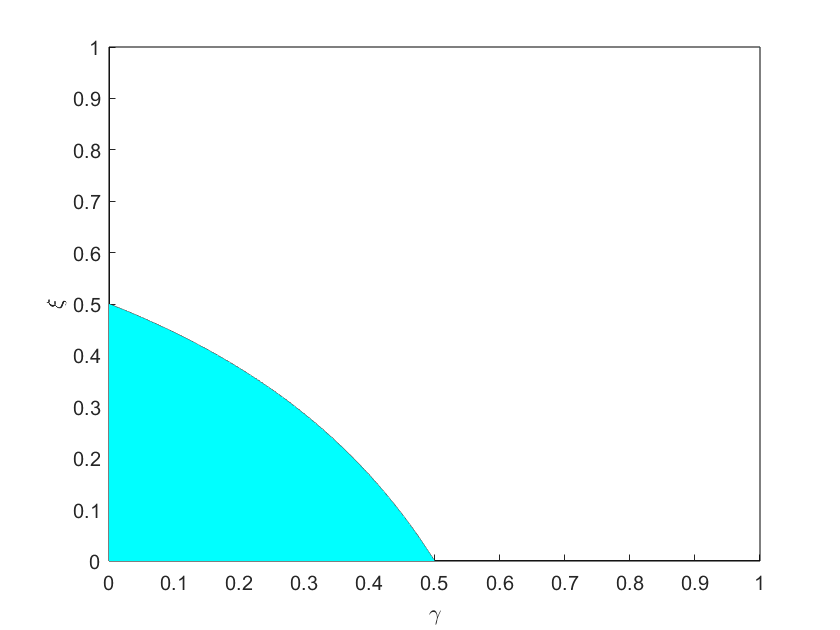}}
%\subfigure []
%{\includegraphics[width=3.9cm,height=2.9cm]%{mix_noise_case_one_state_measure.png}}
\caption{\quad (a) The region of $\Omega^{\rm noisy}(\mu, \beta)$. (b) The region $\Omega^{\rm noisy}(\alpha, \delta)$. (c)  The region $\Omega^{\rm noisy}(\gamma^{amp}, \xi^{amp})$.}
\centering \label{f4}
\end{figure}

Next we devote to analyzing the case of combinations of two different kinds of noise parameters at least.

{\bf Case 2.}  (Mixed consistent noises from entanglement generation and measurements) $\alpha_i=\alpha,\delta_i=\delta$ and $\mu_i=\mu,\beta_i=\beta$, $\forall i = 1, 2,..., n$, other noises vanish. For visibility, we further assume that $\alpha=\delta, \mu=\beta$. The inequality
criterion (\ref{cnsv}) becomes $\mu^2\alpha\sqrt{\alpha^2+1}>1$. The corresponding region $\Omega^{\rm noisy}(\alpha=\delta, \mu=\beta)$ is plotted as the blue region in Fig. 3 (a). One can calculate that under the noisy level $\alpha=\delta=0.95$ and $\mu=\beta=0.98$ in the $\Omega^{\rm noisy}(\alpha=\delta, \mu=\beta)$, the linear network non-7-local correlation at most can be demonstated.

Similarly, we can deal with the following four cases:
{\bf Case 3.} $\alpha_i=\alpha=\delta_i=\delta$, and $\gamma_i^{amp}=\gamma^{amp}=\xi_i^{amp}=\xi^{amp}$, $\forall i = 1, 2,..., n$, other noises vanish.
{\bf Case 4.}
$\alpha_i=\alpha=\delta_i=\delta$, and $\gamma_i^{ph}=\gamma^{ph}=\xi_i^{ph}=\xi^{ph}$, $\forall i = 1, 2,..., n$, other noises vanish.
{\bf Case 5.}
$\mu_i=\mu=\beta_i=\beta$,  and $\gamma_i^{amp}=\gamma^{amp}=\xi_i^{amp}=\xi^{amp}$,  $\forall i = 1, 2,..., n$, other noises vanish.
{\bf Case 6.}
$\mu_i=\mu=\beta_i=\beta$,  and $\gamma_i^{ph}=\gamma^{ph}=\xi_i^{ph}=\xi^{ph}$, $\forall i = 1, 2,..., n$, other noises vanish. The corresponding infinite-persistency set in the above four cases are plotted in Fig. 3 (b-e), respectively. Comparing the regions in Fig. 3 to those in Fig. 2, roughly speaking, the infinite-persistency regions under mixed noises becomes smaller than those under the single kind of noises. This means many kinds of noises make the stronger decay for non-$n$-local correlation. From Fig. 3 (c) and (e), although single  PD channels can not make the decay of star network non-$n$-local correlations, they can work by combining with other kinds of noises together. Moreover, compared Fig. 3 (b)/(d) to (c)/(e), we can conclude that in the mixed noise case, the amplitude-damping channels can make the corresponding infinite-persistency regions more smaller than phase damping channels.

\begin{figure}\label{twomixednoise}
\center
\subfigure []
{\includegraphics[width=3.9cm,height=2.9cm]{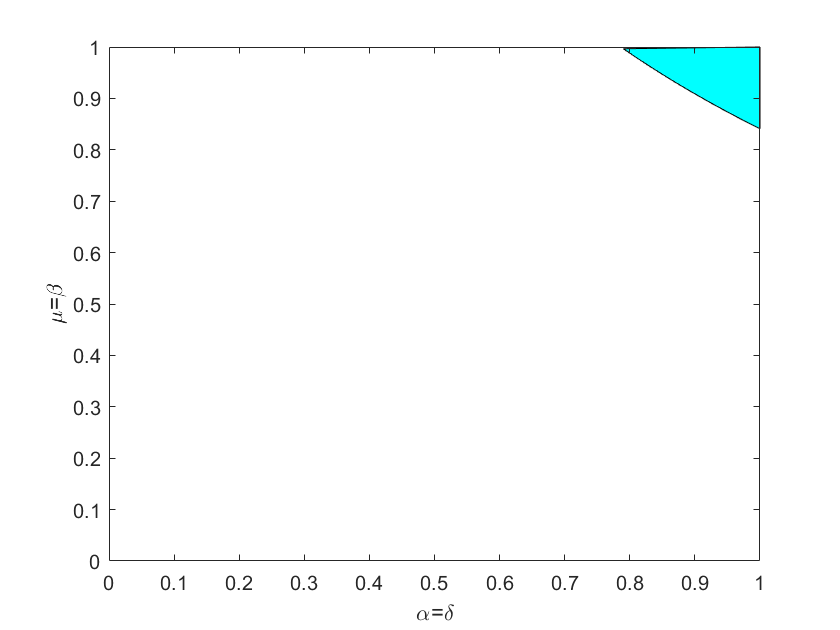}}
\subfigure []
{\includegraphics[width=3.9cm,height=2.9cm]{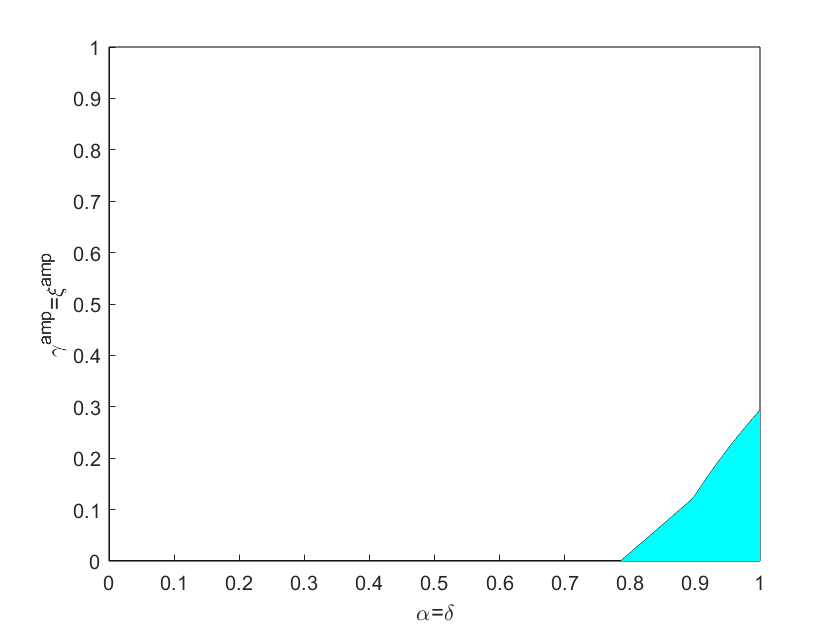}}\\
\subfigure []
{\includegraphics[width=3.9cm,height=2.9cm]{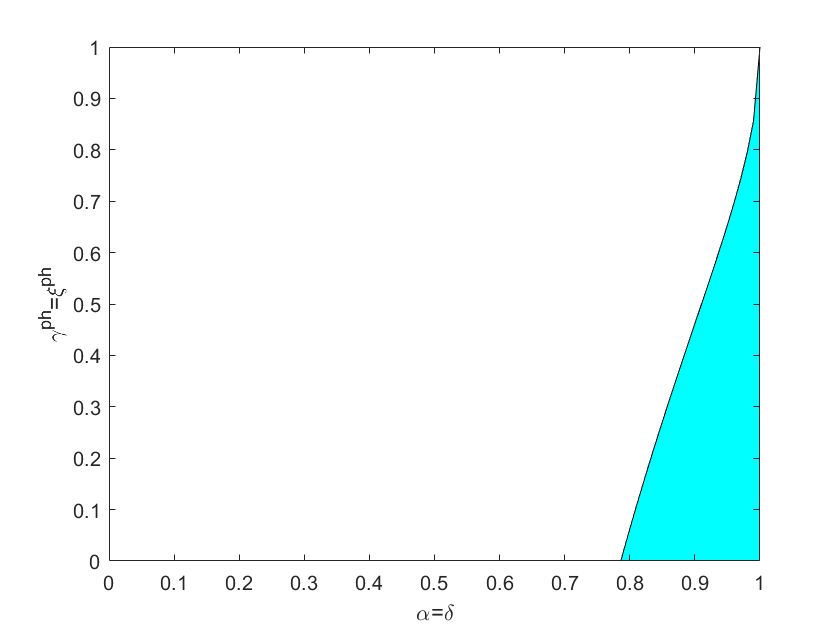}}
\subfigure []
{\includegraphics[width=3.9cm,height=2.9cm]{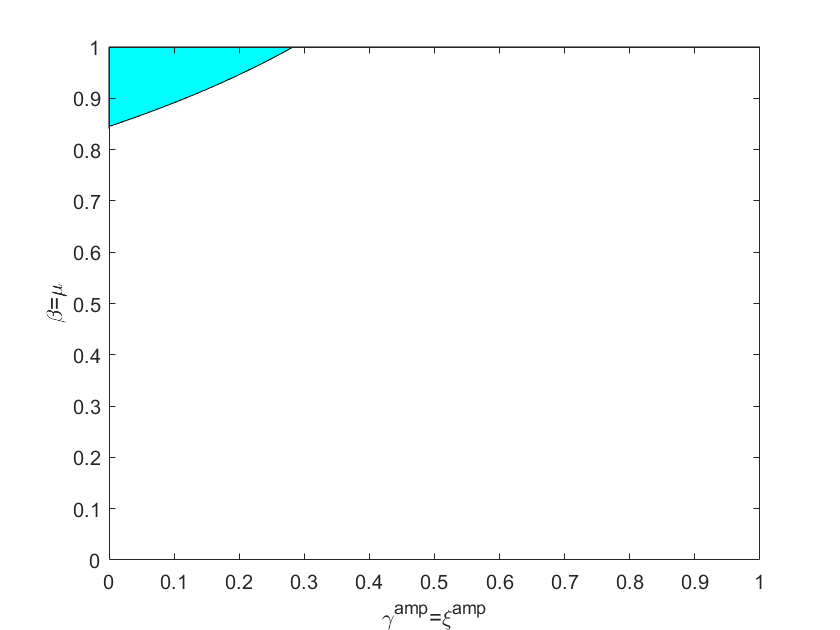}}
\subfigure []
{\includegraphics[width=3.9cm,height=2.9cm]{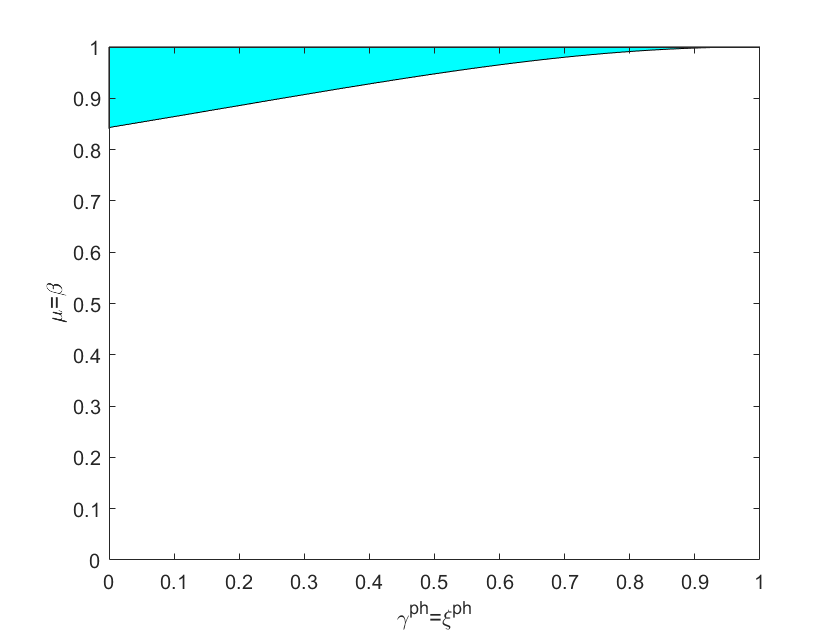}}
\caption{\quad (a) The region of $\Omega^{\rm noisy}(\alpha=\delta, \mu=\beta)$. (b) The region $\Omega^{\rm noisy}(\alpha= \delta, \gamma^{amp}=\xi^{amp})$. (c)  The region $\Omega^{\rm noisy}(\alpha=\delta,\gamma^{ph}= \xi^{ph})$. (d) The region $\Omega^{\rm noisy}(\gamma^{amp}=\xi^{amp},\mu=\beta)$. (e) The region $\Omega^{\rm noisy}(\gamma^{ph}=\xi^{ph},\mu=\beta)$.}
\centering \label{f4}
\end{figure}

{\bf Case 7.}  (Mixture of all consistent noises)
$\alpha_i=\alpha=\delta_i=\delta$, $\mu_i=\mu=\beta_i=\beta$,  and $\gamma_i^{ph}=\gamma^{ph}=\xi_i^{ph}=\xi^{ph},\gamma_i^{amp}=\gamma^{amp}=\xi_i^{amp}=\xi^{amp}$, $\forall i = 1, 2,..., n$. The inequality criterion  (\ref{cchannel1v}) becomes $\beta\mu ({\rm Max}(2\alpha^2 \delta^2 D^{amp},\alpha^2 \delta^2D^{amp}+(\delta D^{amp}+\gamma^{amp}\xi^{amp})^2))^{\frac{1}{2}}>1$ where $D^{amp}=(1-\gamma^{amp})(1-\xi^{amp})$.   The corresponding infinite-persistency set $\Omega^{\rm noisy}(\alpha=\delta, \gamma^{amp}=\xi^{amp},\mu=\beta)$ is plotted as the colored 3D area in Fig. 4 (a). Under the noisy parameter selection in  $\alpha=\delta=0.95, \gamma^{amp}=\xi^{amp}=0.08,\mu=\beta=0.97$, which is in the infinite persistency set $\Omega^{\rm noisy}(\alpha=\delta, \gamma^{amp}=\xi^{amp},\mu=\beta)$, the linear network non-$7$-local correlation can be detected at most.
If amplitude damping channels are replaced by phase camping ones.   The corresponding set $\Omega^{\rm noisy}(\alpha=\delta, \gamma^{ph}=\xi^{ph},\mu=\beta)$ is plotted in Fig. 4 (b).

\begin{figure}\label{twomixednoise}
\center
\subfigure []
{\includegraphics[width=3.9cm,height=2.9cm]{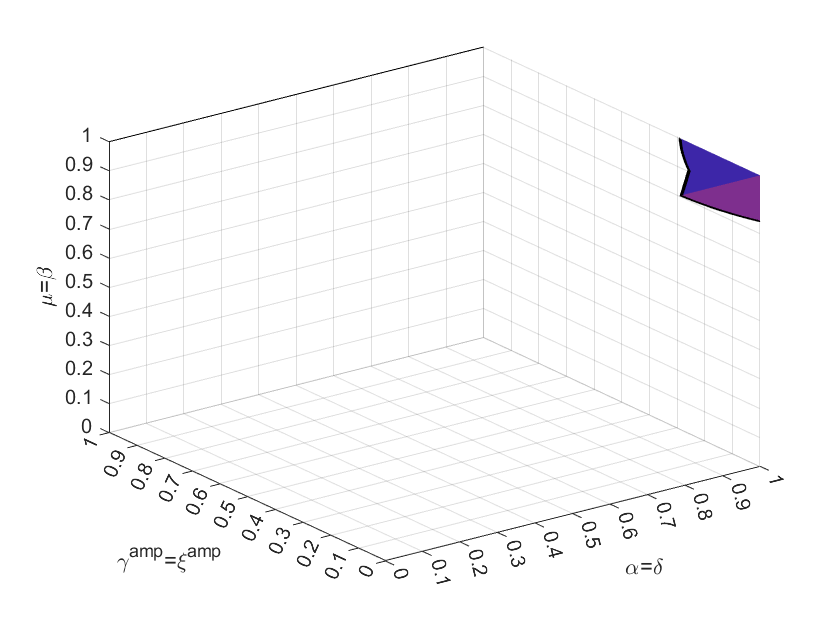}}
\subfigure []
{\includegraphics[width=3.9cm,height=2.9cm]{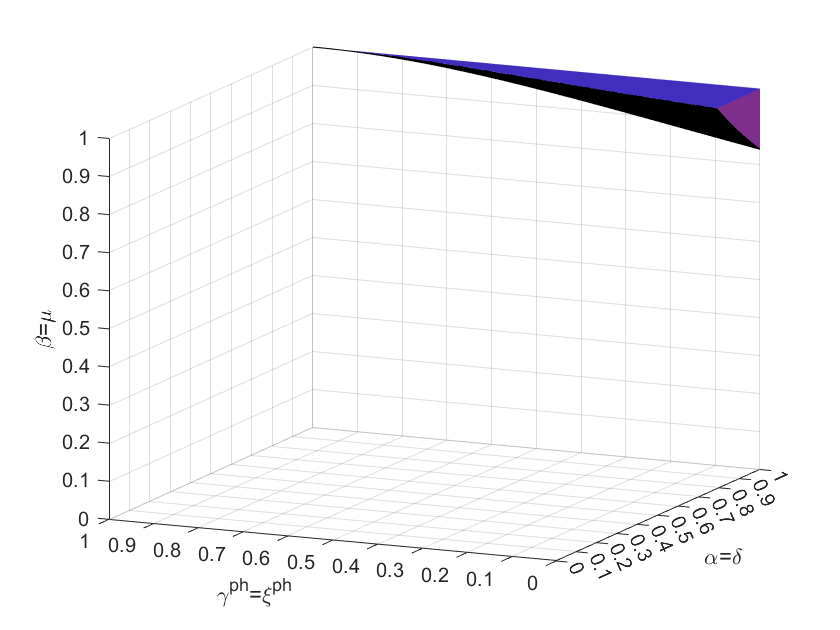}}
\caption{\quad (a) The region of $\Omega^{\rm noisy}(\alpha=\delta, \gamma^{amp}=\xi^{amp},\mu=\beta)$. (b) The region $\Omega^{\rm noisy}(\alpha=\delta, \gamma^{ph}=\xi^{ph},\mu=\beta)$.}
\centering \label{f4}
\end{figure}

\section{Persistency under partially consistent noises}\label{}

In the above section, the infinite persistency of star network non-$n$-local correlations in consistent noisy parameter regions is explored.  However, it still unknown that how the persistency is in the exterior of these regions. For this, here we introduce the topic of partially consistent noises to explore the change pattern of
 the maximal number of sources $n_{\rm max}$ such that non-$n$-local correlation can be demonstrated for arbitrary $n\leq n_{\rm max}$ in the star network. We call the kind of noises with parameters $\alpha_i$ ($i=1,2,...,n$) {\it partially consistent} if there exists a positive integer $k\leq n$, and constants $\alpha$, $\alpha'$ such that $\alpha_i=\alpha$ for all $i= 1,2,...,k$ and $\alpha_j=\alpha'$ for all $j=k+ 1,k+ 2,...,n$. We still divide discussion to the two cases of single and mixed noises. Here it is mentioned that we do not explore all cases for readability, and one can deal with other cases similarly.

\subsection{The case of single noise}

{\bf Case 1.} (state noises only) For $k\leq n$, $\alpha_i=\alpha$ and $\delta_i=\delta$ for  all $i=  1,2,...,k$,  and $\alpha_j=\alpha'$ and $\delta_j=\delta'$ for  all $j=k+ 1,k+ 2,...,n$, other noises vanish. The inequality criterion (\ref{nsv}) becomes
\begin{align}\label{sta1}
 \delta^{\frac{k}{n}} {\delta'}^{\frac{n-k}{n}} \sqrt{\alpha^{\frac{2k}{n}}{\alpha'}^{\frac{2n-2k}{n}}+1 }>1
\end{align}
For visibility, we take $k=1,\alpha=\sqrt{\alpha'},\delta=\sqrt{\delta'}$.
 (\ref{sta1}) becomes
\begin{align}\label{sta11}
{\delta'}^{1-\frac{1}{2n}}\sqrt{{\alpha'}^{2-\frac{1}{n}}+1 }>1
\end{align}
In the case, $n_{\rm max}$ increases like a staircase as $\alpha'$ and $\delta'$  move from 0 to 1 (see Fig. 5 (a)). A common calculation shows that when the noisy parameter pair lies in the exterior of the region $\Omega^{\rm noisy}(\alpha, \delta)$ plotted in Fig. 2 (b), we can obtain the value of $n_{\rm max}$.   For example, when $(\alpha', \delta')=(0.95, 0.7)$, $n_{\rm max}=5$.

\iffalse In the region of noisy parameters $\Omega^{\rm noisy}(\alpha', \delta')=\{(\alpha', \delta'):{\delta'}^{1-\frac{1}{2n}}\sqrt{{\alpha'}^{2-\frac{1}{n}}+1 }>1\}$,  star network non-$n$-local correlation can be demonstrated for arbitrary $n$.
\fi

{\bf Case 2.} (measurement noises only) For $k\leq n$, $\mu_i=\mu$ and $\beta_i=\beta$ for  all $i=  1,2,...,k$,  and $\mu_j=\mu'$ and $\beta_j=\beta'$ for  all $j=k+ 1,k+ 2,...,n$, other noises vanish. (\ref{nsv}) becomes
\begin{align}\label{sta2}
\sqrt{2}\beta^{\frac{k}{n}}{\beta'}^{\frac{n-k}{n}}\mu^{\frac{k}{n}}{\mu'}^{\frac{n-k}{n}}>1
\end{align}
For visibility, we take $k=1,\mu=\sqrt{\mu'},\beta=\sqrt{\beta'}$. (\ref{sta2}) becomes
\begin{align}\label{sta22}
\sqrt{2}{\beta'}^{1-\frac{1}{2n}}{\mu'}^{1-\frac{1}{2n}}>1
\end{align}
Denote by $n_{max}$   the maximal number of $n$ such that (\ref{sta22}) holds true.  $n_{\rm max}$ also increases like a staircase as $\alpha'$ and $\delta'$  move from 0 to 1 (see Fig. 5 (b)).
Similarly when noisy parameters are fixed, for example $(\beta', \mu')=(0.7, 0.95)$, we can calculate the value $n_{\rm max}=3$.

{\bf Case 3.} (PD channel noises) For $k\leq n$, $\gamma^{ph}_i=\gamma$ and $\xi^{ph}_i=\xi$ for  all $i=  1,2,...,k$,  and $\gamma^{ph}_j=\gamma'$ and $\xi^{ph}_j=\xi'$ for  all $j=k+ 1,k+ 2,...,n$, other noises  vanish. (\ref{channel2v}) becomes
\begin{align}\label{sta4}
\sqrt{\left[(1-\gamma)(1-\xi)\right]^{\frac{k}{n}}\left[(1-\gamma')(1-\xi')\right]^{\frac{n-k}{n}}+1}>1
\end{align}
One can see that (\ref{sta4}) always holds true. This says that the same to  the consistent noise case, the non-$n$-local correlations can persist for all $n$ under partially consistent phase damping channels.

\begin{figure}\label{singlenoise}
\center
\subfigure []
{\includegraphics[width=3.9cm,height=2.9cm]{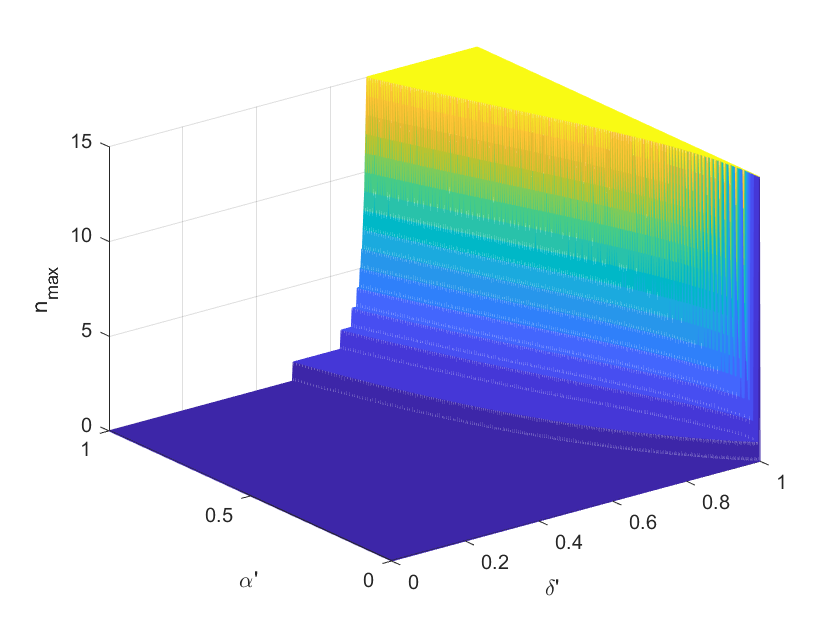}}
\subfigure []
{\includegraphics[width=3.9cm,height=2.9cm]{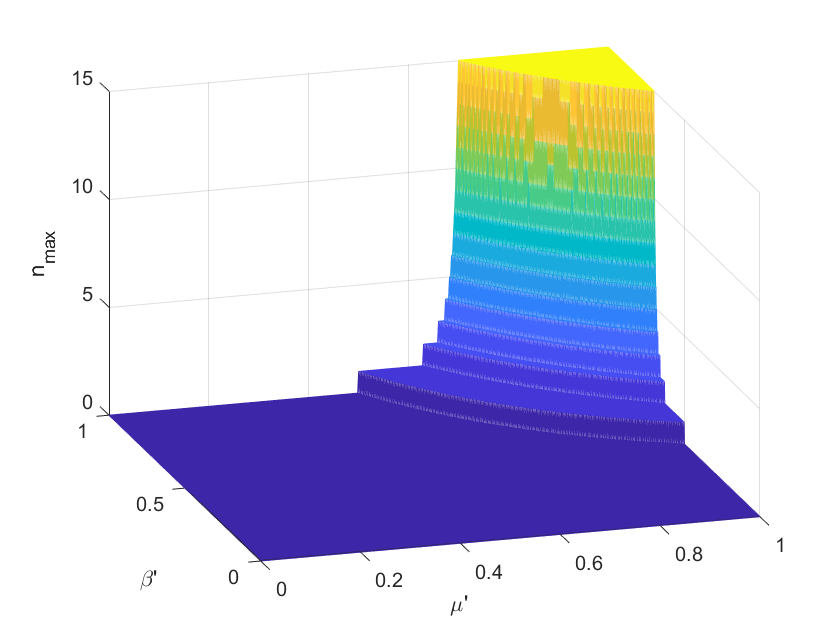}}
\caption{\quad (a) The change pattern of $n_{\rm max}$ under the influence of state noises. (b)  The change pattern of $n_{\rm max}$ under the influence of measurement noises.}
\centering \label{f4}
\end{figure}

\subsection{The case of mixed noises}

{\bf Case 4.} (PD channel and state noises) For $k\leq n$, $\alpha_i=\delta_i=\alpha$ and $\gamma^{ph}_i=\xi^{ph}_i=\gamma$ for  all $i=  1,2,...,k$,  and $\alpha_j=\delta_j=\alpha'$ and $\gamma^{ph}_j=\xi^{ph}_j=\gamma'$ for  all $j=k+ 1,k+ 2,...,n$, other noises vanish. (\ref{channel2v}) becomes
\begin{align}\label{sta7}
\sqrt{\alpha^{\frac{4k}{n}}{\alpha'}^{\frac{4n-4k}{n}}(1-\gamma)^{\frac{2k}{n}}(1-\gamma')^{\frac{2n-2k}{n}}+\alpha^{\frac{2k}{n}}{\alpha'}^{\frac{2n-2k}{n}}}>1
\end{align}
Take $k=1,\alpha=\sqrt{\alpha'},1-\gamma=\sqrt{1-\gamma'}$.
 (\ref{sta7}) becomes
\begin{align}\label{sta77}
\sqrt{{\alpha'}^{4-\frac{2}{n}}(1-\gamma')^{2-\frac{1}{n}}+{\alpha'}^{2-\frac{1}{n}}}>1
\end{align}
We define $n_{max}$ as the maximal number of $n$ such that (\ref{sta77}) holds true. $n_{max}$ increases like a staircase as $\gamma'$ moves from 0 to 1 and $\alpha'$ moves from 1 to 0 (see Fig. 6 (a)).

{\bf Case 5.} (PD channel and measure noises)  For $k\leq n$, $\mu_i=\beta_i=\mu$ and $\gamma_i^{ph}=\xi_i^{ph}=\xi$ for  all $i=  1,2,...,k$,  and $\mu_j=\beta_j=\mu'$ and $\gamma_j^{ph}=\xi_j^{ph}=\xi'$ for  all $j=k+ 1,k+ 2,...,n$, other noises vanish. (\ref{channel2v}) becomes
\begin{align}\label{sta9}
\mu^{\frac{2k}{n}}{\mu'}^{\frac{2n-2k}{n}}\sqrt{(1-\xi)^{\frac{2k}{n}}(1-\xi')^{\frac{2n-2k}{n}}+1}>1
\end{align}
For visibility, we take $k=1,\mu=\sqrt{\mu'},1-\xi=\sqrt{1-\xi'}$.
 (\ref{sta9}) becomes
\begin{align}\label{sta99}
{\mu'}^{2-\frac{1}{n}}\sqrt{(1-\xi')^{2-\frac{1}{n}}+1}>1
\end{align}
$n_{\rm max}$'s change pattern is plotted in Fig. 6 (b).
\begin{figure}\label{twomixednoise}
\center
\subfigure []
{\includegraphics[width=3.9cm,height=2.9cm]{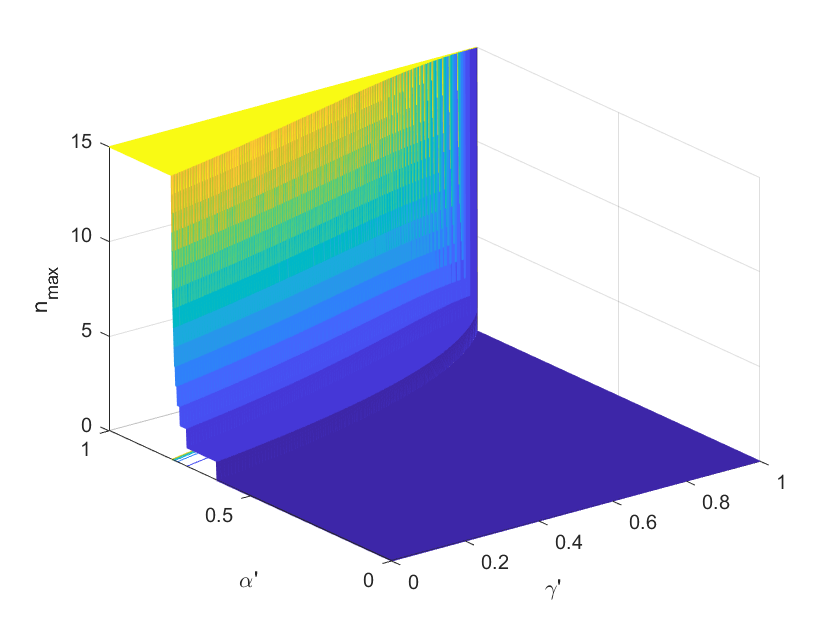}}
\subfigure []
{\includegraphics[width=3.9cm,height=2.9cm]{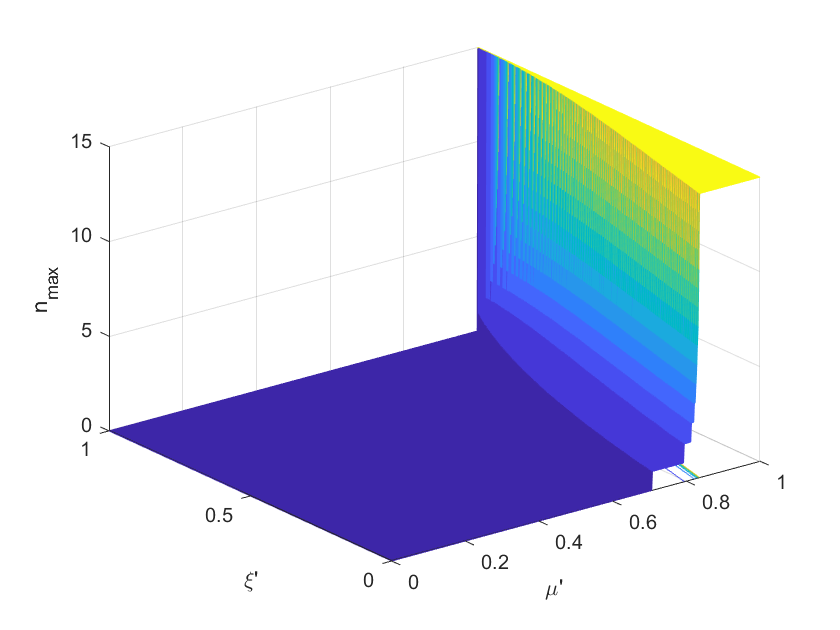}}
\caption{\quad (a)  The change pattern of $n_{\rm max}$ under PD channel and state noises. (b) The change pattern of $n_{max}$ under PD channel and measure noises.}
\centering \label{f4}
\end{figure}

Finally, we list some special values of the source number $n$ under partial consistent noises. It means that the persistency degree can always be calculated when the partially consistent noisy parameters are fixed. While we compare the difference of  persistency source number (PSN) between linear and star network non-multi-local correlations.

\begin{table*}[htbp]
    \centering
    \caption{Persistency of non-$n$-locality in star-networks when $k=1$}
    \label{tab}
    \begin{tabular}{cccc}
       \hline\hline\noalign{\smallskip}
        \textbf{Type of noise} & \textbf{Noise parameters} & \textbf{Star PSN} & \textbf{Linear PSN} \\
        \noalign{\smallskip}\hline\noalign{\smallskip}
        state noises only & $(\alpha',\delta')=(0.95,0.7)$ & 5 & 2 \\
        measurement noises only & $(\beta',\mu')=(0.83,0.83)$ & 7 & 3\\
        PD channel and state noises & $(\alpha',\gamma')=(0.83,0.3)$ & 4 & 2\\
        PD channel and measure noises & $(\mu',\xi')=(0.9,0.35)$ & 4 & 2\\
     \noalign{\smallskip}\hline\hline
    \end{tabular}
\end{table*}

\section{\textbf{Conclusions and discussions}} \label{sec5}

Here we discover that the star network non-$n$-local correlations can offer a  superiority on resisting consistency noises, where they can persist for arbitrary large $n$ in partial or even global regions of consistency noise parameters. Not only that, we propose the topic of partially consistent noises to explore exhaustively the change pattern of the maximal number of $n$ that non-$n$-local correlation can persist under various kinds of noises. Based on our  observations, we conclude that as a kind of noncyclic networks,  the star network can offer the ``strongest" non-$n$-local correlation among noncyclic networks, including linear networks so on.
This suggest that star-type structures should be more available in building quantum networks.

The papre is focused on the noncyclic network case. The challenging further work is to dealt with the case of arbitrary networks with cycles.
Moreover, the several kinds of noises are from independent sources, measurements and channels, which is far from sufficient from experimental perspectives.
It would be more interesting to analyze more errors from physical implementation.

\section*{Acknowledgement}
This work is supported by the National Natural Science Foundation of China under Grants No.12271394, the Key Research and Development Program of Shanxi Province under Grant No.202102010101004.

\section{Appendix}

\begin{widetext}
{\rm Proof of Theorem \ref{th1}}.

According to assumption, a common calculation show that
\begin{align*}
 \langle A_{k_i}^i \otimes B_{0}^i \rangle _{\rho}^{noisy}
&=\sum\limits_{k_i = 0}^1 \sum\limits_{j = 0}^1  \sum\limits_{h = 0}^1 (-1)^{j+h} Tr\left [(P_{i,k_i,j}^{noisy}\otimes M_{i,0,h}^{noisy} )\rho \right ]\\
&=\sum\limits_{k_i = 0}^1  \sum\limits_{h = 0}^1 (-1)^{h} Tr\left\{ \left[(P_{i,k_i,0}^{noisy}-P_{i,k_i,1}^{noisy})\otimes M_{i,0,h}^{noisy} \right]\rho \right\}\\
&=\mu_i \sum\limits_{k_i = 0}^1  \sum\limits_{h = 0}^1 (-1)^{h} Tr\left\{ \left[(P_{i,k_i,0}^{ideal}-P_{i,k_i,1}^{ideal})\otimes M_{i,0,h}^{noisy} \right]\rho \right\}\\
&=\mu_i \sum\limits_{k_i = 0}^1 \sum\limits_{j = 0}^1  \sum\limits_{h = 0}^1 (-1)^{j+h} Tr\left [(P_{i,k_i,j}^{ideal}\otimes M_{i,0,h}^{noisy} )\rho \right ]\\
&=\mu_i \beta_i \sum\limits_{k_i = 0}^1 \sum\limits_{j = 0}^1  \sum\limits_{h = 0}^1 (-1)^{j+h} Tr\left [(P_{i,k_i,j}^{ideal}\otimes M_{i,0,h}^{ideal} )\rho \right ]\\
&=\mu_i \beta_i \langle A_{k_i}^i \otimes B_{0}^i \rangle _{\rho}^{ideal}
\end{align*}

The similar calculations follows that
\begin{align*}
\mathcal{S}_{\rm star}^{\rm noisy}
&=|I|^{1/n}+|J|^{1/n}\\
&=\left | \frac{1}{2}\sum\limits_{{x_1}...{x_n}} {\left\langle A^{1}_{x_1} ...A^{n}_{x_n}B_0 \right\rangle ^{noisy} }\right|^{1/n}+\left|\frac{1}{2}\sum\limits_{{x_1}...{x_n}}(-1)^{\sum\nolimits_i x_i} {\left\langle A^{1}_{x_1} ...A^{n}_{x_n}B_1 \right\rangle ^{noisy}}\right|^{1/n}\\
&=\left|\prod\limits_{i = 1}^n  \frac{1}{2}(\langle A_0^i B_0^i\rangle^{noisy}+\langle A_1^i B_0^i\rangle^{noisy})\right|^{1/n}+\left|\prod\limits_{i = 1}^n \frac{1}{2}(\langle A_0^i B_1^i\rangle^{noisy}-\langle A_1^i B_1^i\rangle^{noisy})\right|^{1/n}\\
&=\left| \prod\limits_{i = 1}^n  \frac{1}{2} \mu_i \beta_i (\vec a^i_0+\vec a^i_1) T_{\rho_{A_iB}} \vec b^i_0\right|^{1/n}+\left| \prod\limits_{i = 1}^n  \frac{1}{2} \mu_i \beta_i (\vec a^i_0 - \vec a^i_1) T_{\rho_{A_iB}} \vec b^i_1 \right|^{1/n}\\
\end{align*}
Analogous to  the proof Theorem 3 in \cite{10}, we have that the maximal value of $\mathcal{S}_{\rm star}^{\rm noisy}$ equals to
\begin{align*}
(\Pi_{i=1}^n\mu_i\beta_i)^{1/n}\sqrt{(\Pi_{i=1}^nt_1^{A_i})^{1/n} +(\Pi_{i=1}^nt_2^{A_i})^{1/n} }
\end{align*}
Where $t_1^{A_i}$and $t_2^{A_i}$are the two greatest eigenvalues of the matrix $T_{\varrho_{{A_i} B}}^T T_{\varrho_{{A_i} B}}$ with $t_1^{A_i} \geq t_2^{A_i}$.

\end{widetext}

\end{document}